\documentclass[a4paper, 11pt]{article}

\usepackage{amsmath, amssymb, amsthm}
\usepackage{xcolor}
\usepackage[shortlabels]{enumitem}
\usepackage{tikz}
\usetikzlibrary {arrows}
\usepackage{cite}

\usepackage{algpseudocode}
\usepackage{algorithm}
\usepackage{geometry}
\newcommand{\abs}[1]{\left\lvert #1 \right\rvert}

\newtheorem{theorem}{Theorem}%[section]
\newtheorem{lemma}[theorem]{Lemma}
\newtheorem*{theorem*}{Theorem}
\newtheorem*{lemma*}{Lemma}
\newtheorem{definition}{Definition}

\usepackage{pifont} % for checkmarks and crosses
\newcommand{\cmark}{\ding{51}}   % ✓
\newcommand{\xmark}{\ding{55}}   % ✗
\usepackage[framemethod=tikz]{mdframed}

\usepackage{mathtools}
\mathtoolsset{showonlyrefs}

\usepackage{thmtools, thm-restate}

\newcommand{\IndState}{\State\hspace{\algorithmicindent}}
\newcommand{\IndIndState}{\IndState\hspace{\algorithmicindent}}
\newcommand{\IndIndIndState}{\IndIndState\hspace{\algorithmicindent}}
\newcommand{\IndIndIndIndState}{\IndIndIndState\hspace{\algorithmicindent}}

\newcommand{\amp}{\mathrm{AMP}}
\newcommand{\ho}{\mathrm{HO}}
\newcommand{\flp}{\mathrm{FLP}}
\newcommand{\cfho}{\mathrm{CFHO}}

\newcommand{\sfho}{\mathrm{SFHO}}

\newcommand{\calM}{\mathcal{M}}
\newcommand{\calP}{\mathcal{P}}
\newcommand{\calQ}{\mathcal{Q}}

\newcommand{\eg}{\emph{e.g.}}

\mdfsetup{%
    roundcorner=5pt,
    skipabove=5pt,
    skipbelow=0pt}

\usepackage{hyperref}

\title{Equivalence and Separation between Heard-Of and Asynchronous Message-Passing Models}

\author{Hagit Attiya\\
    \small{Technion -- Israel Institute of Technology, Israel}
\and
Armando Castañeda\\
\small{Instituto de Matemáticas, Universidad Nacional Autónoma de México, Mexico}
\and
Dhrubajyoti Ghosh\\
\small{Université Paris-Saclay, CNRS, ENS Paris-Saclay, Laboratoire Méthodes Formelles, France}
\and
Thomas Nowak\\
\small{Université Paris-Saclay, CNRS, ENS Paris-Saclay, Laboratoire Méthodes Formelles, France}\\
\small{Institut Universitaire de France, France}
}

\date{}

\begin{document}

\maketitle

\begin{abstract}
We revisit the relationship between two fundamental models of distributed computation:
the asynchronous message-passing model with up to~$f$ crash failures ($\amp_f$)
and the Heard-Of model with up to~$f$ message omissions ($\ho_f$).
We show that for $n > 2f$, the two models are equivalent with respect to the
solvability of colorless tasks, and that for colored tasks the equivalence holds
only when $f = 1$ (and $n > 2$).
The separation for larger $f$ arises from the presence of \emph{silenced processes}
in $\ho_f$, which may lead to incompatible decisions.
The results are proved through bidirectional simulations between $\amp_f$ and $\ho_f$,
using an intermediate model that captures this notion of silencing.
The results extend to randomized protocols against a non-adaptive adversary,
indicating that the expressive limits of canonical rounds are structural rather
than probabilistic.
Together, these results help to delineate where round-based abstractions capture
asynchronous computation, and where they do not.
\end{abstract}


%{\footnotesize
%	\tableofcontents}

%\thispagestyle{empty}
%\newpage
%\setcounter{page}{1}

\section{Introduction}

%% Motivation and Context

%{\itshape
%    Explain the general problem of relating different models of distributed computation.

%    Emphasize that such relationships are central to understanding the expressiveness
%    and power of these models: which problems (decision tasks) can be solved in one but not the other.

%    Mention that simulations between models often provide these relationships,
%    and are essential in comparing asynchrony, synchrony, and fault tolerance.
%}

Distributed computing has been studied through a variety of models that differ
in communication primitives, synchrony assumptions, and failure types.
A natural question is how these models compare in their computational power.
Specifically, we compare the sets of \emph{tasks} solvable in different models.
This is particularly valuable when two models are equivalent,
but one offers a simpler setting for analysis or verification.
Simulations are often used to provide these relationships,
relating models that differ in synchrony or fault tolerance.

%This can also help in understanding their expressiveness and power: which problems (decision tasks in this context) can be solved in one but not the other.
%Over the years there have been attempts to bring together many of these models under a unified and abstract framework.

%% Background on the Models

%Introduce briefly the AMPf model (asynchronous message-passing with up to f crash failures).
%Introduce the HOf model (round-based, synchronous abstraction with message omissions modeled via Heard-Of sets).
%Mention that the Heard-Of model has provided a unified perspective for many benign fault models, but its precise relationship with asynchronous message passing under crash failures remains non-trivial.

The standard \emph{Asynchronous Message-Passing model with $f$ failures} ($\amp_f$)
assumes unbounded message delivery times and allows up to $f$ process crashes.
Many protocols in this model proceed in \emph{rounds}: each process broadcasts
a message tagged with its current round number and waits for $n - f$ messages
from that round before advancing to the next.
This round-based behavior is formalized in the \emph{Heard-Of model}~\cite{HeardOf},
which describes communication through \emph{Heard-Of sets}, the sets of processes
from which each process receives a message in a given round.
The variant with $f$ message omissions, denoted $\ho_f$, corresponds
to executions in which every process hears from at least $n - f$ processes per round.

The $\ho_f$ model is well suited to algorithm design and verification,
partly because its executions form a compact set in a topology based on common
prefixes (proved in Appendix \ref{sec:compactness}).
Compactness implies that every infinite execution has finite,
well-defined prefixes that can be reasoned about independently.
As a result, correctness can be expressed through safety predicates and
locally finite execution trees,
avoiding the need to handle liveness properties explicitly as in $\amp_f$.
This makes the $\ho_f$ model particularly amenable to both
manual proofs and automated verification.

%% Known Results and Open Questions:

%{\itshape
%    Summarize known equivalences (e.g., synchrony–asynchrony unifications,
%    partial results such as those by Gafni, Charron-Bost and Schiper).

%    Point out that while equivalence for colorless tasks (like consensus)
%    has been studied, the relationship for colored tasks and
%    for larger fault ratios is less clear.

%    Mention the lack of direct correspondence for all parameter ranges.
%}

The precise relationship between $\ho_f$ and $\amp_f$ has remained unclear.
In particular, it is not known whether canonical rounds are fully general,
or whether all tasks solvable in $\amp_f$ can be solved also in $\ho_f$.
Indeed, not every protocol in $\amp_f$ conforms to the canonical-round pattern.
For example, the renaming task~\cite{attiya_renaming_amp_1990},
is solved in $\amp_f$ by a protocol in which each process maintains and
shares the set of initial names it knows
and decides once it has received $n - f - 1$ additional copies of this set.
Here, decisions depend on the number of identical messages,
rather than on the number of messages associated with a particular round.
This example highlights that not all asynchronous protocols are expressed in terms of
canonical rounds, motivating a closer examination of the precise boundary
between $\amp_f$ and $\ho_f$.

The most relevant prior work is due to Gafni and Losa~\cite{gafni2024timehealer},
who showed that $\ho_1$ and $\amp_1$
solve the same set of \emph{colorless} tasks for $n > 2$.
(Roughly speaking, in a colorless task, like consensus,
processes can adopt each other's decisions.)
However, to the best of our knowledge, the case of colored tasks and
larger fault thresholds ($f > 1$) has not been explored.

In this work, we complete the characterization of the relative computational power
of $\ho_f$ and $\amp_f$ \emph{for arbitrary $f$}, considering both colorless and colored tasks.

%% Our Contribution:
%{\itshape
%    A bidirectional simulation between HOf and AMPf, yielding conditions for equivalence.
%    The introduction of CFHOf and SFHOf as intermediate models to bridge the two.

%    The main theorems showing:
%    For colorless tasks, HOf and AMPf coincide when f < n/2.
%    For colored tasks, equivalence holds for f = 1, but not for f > 1.
%}

To relate $\ho_f$ and $\amp_f$, we introduce an intermediate model
based on the notion of silenced processes.
A process in $\ho_f$ is considered \emph{silenced} when its influence
eventually stops propagating: only a proper subset of processes continue
to hear from it, directly or indirectly.
The \emph{Silenced-Faulty} model, $\sfho_f$, treats such processes as faulty,
while in $\ho_f$ they are still required to decide.
Thus, establishing the equivalence between $\ho_f$ and $\amp_f$ reduces
to proving the equivalence between $\ho_f$ and $\sfho_f$.
Whether this holds depends on the number of failures and
on whether the task is colorless or colored.

We use bidirectional simulations to prove the following results (see Table~\ref{table:summary}):
\begin{itemize}
	\item For $n > 2f$, $\amp_f$ and $\ho_f$ solve the same set of colorless tasks
	      (Theorem~\ref{thm:colorless}).

	\item For $n > 2$ and $f \le 1$, $\amp_f$ and $\ho_f$ solve the same set of colored tasks (Theorem~\ref{thm:colored_f=1}).

    \item For $f > 1$ and $n > 2f$, the renaming task, which is colored, is solvable in $\amp_f$ but \emph{not} in $\ho_f$ (Theorem~\ref{thm:colored_f>1}).
\end{itemize}

\begin{table}[tb]
	\begin{center}
		\begin{tabular}[]{|c|c|c|c|c|c|}
			\hline
			                    & \shortstack{$0\leq f \leq 1$                                                                                                                                                                                                                           \\colorless} & \shortstack{$0\leq f \leq 1$\\colored}                              & \shortstack{$f > 1$ and $n > 2f$                     \\colorless} & \shortstack{$1 < f < n/2$ \\ colored} & $f \geq n/2$ \\
			\hline
			$\ho_f$ in $\amp_f$ & \cmark{\scriptsize Lemma~\ref{lem:HOinAMP}}    & \cmark{\scriptsize Lemma~\ref{lem:HOinAMP}}       & \cmark{\scriptsize Lemma~\ref{lem:HOinAMP}}     & \cmark{\scriptsize Lemma~\ref{lem:HOinAMP}}       & \cmark{\scriptsize Lemma~\ref{lem:HOinAMP}} \\
			\hline
			$\amp_f$ in $\ho_f$ & \cmark{\scriptsize \cite{gafni2024timehealer}} & \cmark{\scriptsize Theorem~\ref{thm:colored_f=1}} & \cmark{\scriptsize Theorem~\ref{thm:colorless}} & \xmark{\scriptsize Theorem~\ref{thm:colored_f>1}} & unknown                                     \\
			\hline
		\end{tabular}
	\end{center}
	\caption{Conditions under which tasks solvable in $\ho_f$
		are also solvable in $\amp_f$ and vice versa.}
	\label{table:summary}
\end{table}

%% Technical Highlights

%{\itshape
%    Describe the key methodological idea:
%    Construction of simulation systems linking the two models.
%    The concept of silenced processes in SFHOf and its role in distinguishing model power.
%    Perhaps say on how this relates to prior simulations
%    (e.g., [Attiya Welch 2004], [Gafni Kuznetsov 2011]).
%}

We also show that the results carry over to randomized
protocols with private local coins and a non-adaptive adversary, i.e.,
one whose behavior is fixed independently of the protocol’s random choices;
see Theorems~\ref{thm:colorless_rand}, \ref{thm:colored_f=1_rand},
and~\ref{thm:colored_f>1:randomized}.
With an adaptive adversary, the separation result still holds because
they are stronger than non-adaptive adversaries;
however, the equivalence results remain open.

%These results are summarized in Table~\ref{table:random}.
%
%\begin{table}
%\begin{center}
%    \begin{tabular}[]{|c|c|c|c|c|c|}
%        \hline
%        & \shortstack{$0\leq f \leq 1$                                                                                                                                                                                                                           \\colorless} & \shortstack{$0\leq f \leq 1$\\colored}                              & \shortstack{$1 < f < n/2$                       \\colorless} & \shortstack{$1 < f < n/2$ \\ colored} & $f \geq n/2$ \\
%        \hline
%        Non-adaptive adversary & \cmark{\scriptsize Theorem~\ref{thm:colorless_rand}}                   & \cmark{\scriptsize Theorem~\ref{thm:colored_f=1_rand}}                      & \cmark{\scriptsize Theorem~\ref{thm:colorless_rand}}                    & \xmark{\scriptsize Theorem~\ref{thm:colored_f>1:randomized}}                      & \qmark \\
%        \hline
%        Adaptive adversary & \qmark & \qmark & \qmark & \xmark{\scriptsize Theorem~\ref{thm:colored_f>1:randomized}} & \qmark                                      \\
%        \hline
%    \end{tabular}
%\end{center}
%\caption{Relations between $\ho_f$ and $\amp_f$, for randomized algorithms.}
%\label{table:random}
%\end{table}

\paragraph*{Additional related work:}

A long line of work has sought to bridge asynchronous and round-based
message-passing models by placing them within a common framework for
reasoning about protocols and impossibility results.
The Heard-Of model of Charron-Bost and Schiper~\cite{HeardOf},
Gafni’s round-by-round fault detectors~\cite{gafni_round-by-round_1998},
and the unification framework of Herlihy, Rajsbaum, and
Tuttle~\cite{HerlihyRT98} are central to this direction.
The problem of finding a round-based message-passing model equivalent
to the read-write wait-free model was addressed by
Afek and Gafni~\cite{afek2012asynchronysynchrony}.
More recently, it was shown that canonical rounds are not fully general,
when failures are Byzantine~\cite{AttiyaFW25}.
Our work contributes to this line by refining the comparison between
asynchronous and round-based computation and pinpointing
conditions under which equivalence holds.

Another body of research studies how asynchronous message-passing models
can be characterized directly in round-based terms.
Shimi, Hurault, and Quéinnec~\cite{ShimiHQ2021} introduce \emph{delivered predicates}
and show how they generate Heard-Of predicates for asynchronous models,
providing a formal mechanism for deriving round-based behavior from
asynchronous assumptions.
Our approach differs in that we focus on \emph{task solvability}
and use the notion of \emph{silenced processes}
to identify where the Heard-Of abstraction $\ho_f$ diverges from $\amp_f$.

The Heard-Of abstraction has also proved valuable for the verification
of distributed algorithms.
Several works, including~\cite{ShimiHQ2021,DebratMerz12,BalasubramanianW2020,DamianDMW19,dragoi_logic-based_2014}, have used it to reason about
correctness within purely safety-based frameworks.
Balasubramanian and Walukiewicz~\cite{BalasubramanianW2020} formalized
this connection, showing that verification in the Heard-Of model can be
reduced to reasoning over locally finite trees.
Because $\ho_f$ is defined by safety predicates and its executions form
locally finite trees, automated verification is often more tractable
than in $\amp_f$, which also requires liveness properties.
%\todo{find an older citation?}
The compactness result we establish for $\ho_f$ provides a theoretical
explanation for this practical advantage.

Beyond verification, other comparisons between message-passing models
have focused on the relationship between failure assumptions.
Analyses such as Gafni’s round-by-round fault detectors~\cite{gafni_round-by-round_1998}
investigate reductions between crash and omission failures,
highlighting algorithmic reductions between models.
Our perspective differs: we study \emph{computational equivalence},
namely, whether two models can solve exactly the same set of tasks,
rather than the existence of specific simulation-based transformations.
This shift in focus clarifies the boundaries between asynchronous
and round-based computation.
In particular, our results refine the Heard-Of framework
by identifying when the abstraction of canonical rounds captures
the behavior of asynchronous message passing, and when it does not.

\section{Models and preliminaries} \label{sec:models}

We consider message-passing models where processes modeled as deterministic automata with infinite state spaces communicate by broadcasting and receiving messages.
A \emph{protocol} is a set of $n$ processes $\mathcal{P} = \left\{ p_1, \dots, p_n \right\}$.
We assume our protocols to be \emph{full-information}, where each process always broadcasts its current \emph{view}, which is initially its initial state and afterwards its complete local history.
A configuration of the system consists of the local states of all the processes and also the state of the environment, e.g., messages in transit, etc.
A \emph{schedule} describes the order in which processes broadcast and receive messages, and which processes' messages are obtained during each receive.
An initial configuration $I$ and a schedule $S$ determine a \emph{run} $\alpha(I, S)$ of the system which is an alternating sequence of configurations and broadcast or receive events.
Each model describes a set of possible schedules and for each schedule,
a set of processes that are \emph{faulty}.

\subsubsection*{The $\amp_f$ model}

In the \emph{Asynchronous Message-Passing Model with $f \geq 0$ process failures}
($\amp_f$), there is no fixed upper bound on the time it takes for
a message to be delivered nor on the relative speeds of the processors.
Processes proceed in atomic steps, where they execute a receive, then a local computation and finally a broadcast.
In a broadcast event, denoted as $\operatorname{amp-bc}_i(m)$, process $p_i$ broadcasts message $m$ to all processes including itself.
In a receive event, denoted $\operatorname{amp-recv}_i(m)$, process $p_i$ receives message $m \in M \cup \left\{ \bot \right\}$ where $M$ is the message domain and $\bot$ denotes ``no message".

There are multiple definitions for $\amp$; we present one that
is equivalent to the classical definition by Fischer, Lynch,
and Patterson~\cite{FLP} (see Appendix~\ref{sec:AMP=FLP}).
A process that has infinitely many steps in a run is said to be \emph{non-faulty}, and \emph{faulty} otherwise.
There can be at most $f \geq 0$ faulty processes.
Messages broadcast by a non-faulty process are received by all non-faulty processes (Non-faulty Liveness).
All messages \emph{but} the last broadcast by a faulty process are received by all non-faulty processes (Faulty Quasi-Liveness).
Every message received by a process was previously sent to the process (Integrity).
No message is received more than once at any process (No Duplicates).

\subsubsection*{The $\ho_f$ model} \label{HO_f}

The Heard-Of model~\cite{HeardOf} is a round-based model.
Here we consider one of its instances, the \emph{Heard-Of model with $f \geq 0$ message omissions}.
In each round, every process does a broadcast, then a receive, and finally a local computation.
In a broadcast event, denoted $\operatorname{ho-bc}_i(m)$, process $p_i$ broadcasts $m$ to all processes including itself.
In a receive event, denoted $\operatorname{ho-recv}_i(M')$, it gets a subset $M'$ of the messages sent to it during the round.
Messages missed in a round are lost forever.

Conceptually, processes proceed in lockstep; to represent the computation in the form of a sequence, we impose a total order on the events that is consistent with the order of events at each process, e.g., the round-robin pattern \cite[Section 11.1]{attiyawelch}.
In each round, a process may lose at most~$f$ messages from other processes;
it always hears from itself (Weak Liveness).
Every message received by a process in round $k$ must have been broadcast in round $k$ by some process (Integrity).
A receive event can contain at most one message from each neighbor (No Duplicates).
We remark that no process is defined to be faulty.

\subsubsection*{Tasks}

In a run of a protocol,
every process $p_i$ has a non-$\bot$ \emph{input} value
in its initial state and \emph{decides} on
an \emph{output} value \emph{at most once}.
If a process does not decide, its output is denoted by $\bot$.

A \emph{task} is a tuple $T = (\mathcal{I}, \mathcal{O}, \Delta)$ where $\mathcal{I}$ is a set of input vectors (one input value for each process), $\mathcal{O}$ is a set of output vectors (one output value for each process), and $\Delta$ is a total relation that associates each input vector in $\mathcal{I}$ with a set of possible output vectors in $\mathcal{O}$.
An output value of $\bot$ denotes an \emph{undecided} process.
We require that if $(I, O) \in \Delta$, then for each $O'$ resulting after replacing some items in $O$ with $\bot$, $(I, O') \in \Delta$.

In a \emph{colorless} task, processes are free to copy the inputs and outputs of other processes.
Formally, let $val(U)$ denote the \emph{set} of non-$\bot$ values in a vector $U$.
In a colorless task, for all input vectors $I, I'$ and all output vectors $O, O'$ such that $(I, O) \in \Delta$, $val(I) \subseteq val(I')$, and $val(O') \subseteq val(O)$, we have $(I', O) \in \Delta$ and $(I, O') \in \Delta$.

A task is \emph{colored} if is not colorless.

A protocol $\calP$ \emph{solves} $T$ in model $\mathcal{M}$, if in every run $R$ of $\calP$ with input vector $I$ and output vector $O$, we have (1) $(I, O) \in \Delta$ and (2) $O[i] = \bot$ only if $p_i$ is faulty in $\mathcal{M}$ in $R$.
Finally, task $T$ is \emph{solvable} in model $\mathcal{M}$ if there exists a protocol $\calP$ that solves $T$ in $\mathcal{M}$.

\subsubsection*{Simulations} \label{sec:simulations}

The definition of simulations follows~\cite[Chapter~7]{attiyawelch}.
A simulation of model $\mathcal{M}_2$ in model $\mathcal{M}_1$ consists of three layers:
(1) $n$ \emph{simulated processes} $p_1, \dots, p_n$,
(2) $n$ \emph{simulating machines} $P_1, \dots, P_n$, with machine~$P_i$ assigned to simulated process~$p_i$, and
(3) the communication system of~$\mathcal{M}_1$ (see Figure~\ref{fig:simulation_system}).
Each simulated process~$p_i$ interacts with its simulating machine~$P_i$ via~$\calM_2$-communication primitives, as if the communication system is that of~$\calM_2$.
The simulating machines interact with each other through~$\calM_1$ using~$\calM_1$-communication primitives.
Transitions between states in simulated processes and simulating machines are triggered by the occurrence of events in~$\calM_1$ and~$\calM_2$.
The occurrence of an event between~$p_i$ and~$P_i$ entails a transition in both~$p_i$ and~$P_i$.
A simulation is described by specifying the automaton for the simulating machines.

	\begin{figure}[bt]
		\centering
		\begin{minipage}[]{0.4\textwidth}
			\centering
			\begin{tikzpicture}
				\draw[->, >=latex] (-0.8, 0) -- (-0.8, -0.8) node[midway, left] {$\calM_2$-bc};
				\draw[->, >=latex] (0.8, -0.8) -- (0.8, 0) node[midway, right] {$\calM_2$-recv};
				\node at (0, 0.4) {simulated process $p_i$};
				\node at (0, -1.1) {simulating machine $P_i$};
				\node at (0, -2.5) {communication system $\calM_1$};
				\draw[->, >=latex] (-0.8, -1.4) -- (-0.8, -2.2) node[midway, left] {$\calM_1$-bc};
				\draw[->, >=latex] (0.8, -2.2) -- (0.8, -1.4) node[midway, right] {$\calM_1$-recv};
			\end{tikzpicture}
			\caption{Simulating $\calM_2$ in $\calM_1$}
			\label{fig:simulation_system}
			%\begin{tikzpicture}
			%  \node at (1.25, 2.55)  {$p_i$};
			%  \node at (1.25, 1.05) {$P_i$};
			%  \fill[color=gray, opacity=0.6] (0.3, 0.7) rectangle (2.2, 1.4);
			%  \fill[color=gray, opacity=0.6] (0.3, 2.2) rectangle (2.2, 2.9);
			%  \draw[->,>=latex] (0.7, 2.2) -- +(0, -0.8);
			%  \draw[->,>=latex] (1.8, 1.4) -- +(0, 0.8);
			%  \draw[->,>=latex] (0.7, 0.7) -- +(0, -0.8);
			%  \draw[->,>=latex] (1.8, -0.1) -- +(0, 0.8);

			%  \node at (-0.2, 0.3) {amp-send};
			%  \node at (-0.5, 1.8) {ho-send};
			%  \node at (2.7, 0.3) {amp-recv};
			%  \node at (3.0, 1.8) {ho-recv};
			%\end{tikzpicture}
		\end{minipage}\hspace{0.8cm}
		\begin{minipage}[]{0.5\textwidth}
			\centering
			\begin{tikzpicture}
				\node at (1.25, 2.55)  {$p_i$};
				\node at (1.25, 1.05) {$P_i$};
				\node at (1.25, 3.2) {$q_i$};
				\draw (0.1, 0.55) rectangle (2.4, 3.0);
				\fill[color=gray, opacity=0.6] (0.3, 0.7) rectangle (2.2, 1.4);
				\fill[color=gray, opacity=0.6] (0.3, 2.2) rectangle (2.2, 2.9);
				\draw[->,>=latex] (0.7, 2.2) -- +(0, -0.8);
				\draw[->,>=latex] (1.8, 1.4) -- +(0, 0.8);
				\draw[->,>=latex] (0.7, 0.7) -- +(0, -0.8);
				\draw[->,>=latex] (1.8, -0.1) -- +(0, 0.8);

                \node at (-0.2, 0.3) {amp-bc};
                \node at (-1.2, 1.8) {ho-bc(internal)};
				\node at (2.9, 0.3) {amp-recv};
				\node at (3.9, 1.8) {ho-recv(internal)};
			\end{tikzpicture}
			\caption{Constructing $q_i$ from $p_i$ and $P_i$}
			\label{fig:simulation_to_reduction}
			%\caption{Constructing protocol from simulating system}
			%\label{fig:simulation_to_reduction}
		\end{minipage}
	\end{figure}

    A \textit{configuration} of the system consists of the states of all the simulated processes and simulating machines and the state of the network.
    A schedule of a simulation is defined similarly to the schedule of a protocol, except it now involves the primitives of both~$\mathcal{M}_1$ and~$\mathcal{M}_2$.
    A \textit{run} of a simulation consists of an initial configuration~$I$ and a schedule~$\alpha$ such that the events in~$\alpha$ are enabled in turn starting from~$I$.

    Given run~$\alpha$, we define the \textit{simulated run} $top(\alpha)$ by restricting the initial configuration of~$\alpha$ to the simulated processes and the schedule of~$\alpha$ to the events of the upper interface.

\section{Simulating $\ho_f$ in $\amp_f$} \label{sec:HOinAMP}

In this section, we show that tasks solvable in $\ho_f$ are solvable in $\amp_f$,
with a simple simulation showing that an $\ho_f$ protocol can be executed in $\amp_f$,
as a consequence of the canonical-round construction technique in $\amp$.
This helps to prove one direction of
Theorems~\ref{thm:colorless} and~\ref{thm:colored_f=1}.

\begin{lemma}%[Folklore]
	\label{lem:HOinAMP}
	Any task that is solvable in $\ho_f$ is also solvable in $\amp_f$ for $0 \le f \le n$.
\end{lemma}

\begin{proof}
	We first simulate a variant of the $\ho_f$ model, called $\cfho_f$ (for \emph{Crashed-Faulty} $\ho_f$), in $\amp_f$ and then show the inclusion of $\ho_f$ in $\cfho_f$.
	The $\cfho_f$ model is identical to $\ho_f$ except that we allow a subset of processes to crash, i.e., have a finite number of events.
    Crashed processes do not need to decide as they are considered faulty processes.
    The number of messages that a non-crashed process does not receive in a round due to process crashes and message omissions is at most $f$.

	Suppose a protocol $\calP$ solves a task $T$ in $\cfho_f$.
	Each round of $\cfho_f$ can be emulated in $\amp_f$
	by waiting for $n-f$ messages before proceeding,
	preserving the causal structure of message delivery.
	We can construct a simulation system that simulates $\cfho_f$ in $\amp_f$
	so that for $1 \le i \le n$,
	(1) the $i$-th simulated process is $p_i$ and
	(2) the $i$-th simulation machine $P_i$ is specified
	by Algorithm~\ref{algo:HOinAMP}.

	\begin{algorithm}[bt]
		\caption{Pseudocode for machine $P_i$}
		\label{algo:HOinAMP}
		\begin{algorithmic}[1]
			%\State Initialize empty list $seen$
			\State Initialize counter $round \gets 0$
            \State Initialize empty sets $received[r]$ for all $r \ge 0$
			\medskip
            \State{When $\operatorname{ho-bc}_i(m)$ occurs:}
			\IndState $round \gets round + 1$
            \IndState Add $m$ to $received[round]$ \Comment{Ensure that $p_i$ hears from itself in $\ho_f$}
            \IndState Enable $\operatorname{amp-bc}_i(\left\langle m, round \right\rangle)$ \label{algo:enable_amp-sends}
			%\IndState{Enable $\operatorname{amp-send}_i(\left\langle m, round \right\rangle, j)$} for all $j \ne i$ \label{algo:enable_amp-sends}
			\medskip
			\State{When $\operatorname{amp-recv}_i(\left\langle m, r \right\rangle)$ occurs:}
			\IndState Add $m$ to $received[r]$
			\medskip
			\State{Enable $\operatorname{ho-recv}_i(received[round])$ when:}
			\IndState $\abs{received[round]} \ge n - f$
		\end{algorithmic}
	\end{algorithm}

    A protocol $\calQ$ solving task $T = (\mathcal{I}, \mathcal{O}, \Delta)$ in $\amp_f$ can be constructed as in Figure~\ref{fig:simulation_to_reduction} by converting the simulation system into a network of automata $\calQ = \left\{ q_1, \dots, q_n \right\}$ such that $q_i$ runs both $p_i$ and the $i$-th simulation machine $P_i$ internally, and interacts with the communication system of $\amp_f$ using amp-bc and amp-recv events.
    The ho-bc and ho-recv events between $p_i$ and $P_i$ become part of $q_i$'s internal computation.

	The simulator $q_i$ works as follows.
	The input of $q_i$ is also used as the input for $p_i$.
    When $p_i$ does an internal $\operatorname{ho-bc}(m)$ for round $r$ to $P_i$, $q_i$ does an $\operatorname{amp-bc}(\left\langle m, r \right\rangle)$ and waits to receive at least $n - f$ messages of the form $\left\langle \cdot, r \right\rangle$,
	and collects their first components in a set $M'$.
	Then $q_i$ performs an internal $\operatorname{ho-recv}(M')$ from $P_i$ to $p_i$ and changes $p_i$'s state accordingly.
	If $p_i$ internally decides an output $o$, then $q_i$ decides $o$ as well.

        %Let $\alpha$ be a run of $\calQ$ and let $\gamma$ be the corresponding run of the simulation.
    Since machine $P_i$ waits to receive $n - f$ messages tagged with its current round counter before incrementing it, $p_i$ receives $n - f$ messages in every round of a simulated run before broadcasting in the next round.
    The Integrity and No Duplicates properties are straightforward to verify and thus the simulated run is locally valid at every process, though it need not be globally valid.
    Indeed, due to the asynchronous nature of $\amp_f$, machine $P_i$ could receive $n - f$ round-$r$ tagged messages much earlier than machine $P_j$ does, meaning $p_i$ could then move on to do a round-$(r + 1)$ ho-bc before $p_j$ has a ho-recv for round $r$.
    However once the local history of events at the processes are correctly interleaved to follow the round-robin pattern, the rearranged simulated run is valid in $\cfho_f$.
    Thus the vector of outputs (which only depends on the local histories) of all the internal~$p_i$ and thus all~$q_i$ is valid for the input vector.

    We now claim that if $q_i$ is non-faulty in a run of $\amp_f$ then $p_i$ does not crash in the corresponding simulated run.
    To see this, let $\calQ_{\mathrm{NF}}$ be the set of processes in $\calQ$ that are non-faulty in the run of $\amp_f$.
    %We will show that if $q_i \in \calQ_{NF}$ then $p_i$ is non-crashed in $\sigma$.
    Let $q_i \in \calQ_{\mathrm{NF}}$.
    In round 0, $q_i$ does an internal ho-bc from $p_i$ and does a corresponding round-0 tagged amp-bc.
    As $\abs{\calQ_{\mathrm{NF}}} \ge n - f$, there are $n - f$ such amp-bc events.
    Thus $q_i$ receives $n - f$ round-0 tagged messages in $\amp_f$ and enables a ho-recv from $P_i$ to $p_i$; in turn, $p_i$ enables a ho-bc for the next round.
    Thus for every $q_i \in \calQ_{\mathrm{NF}}$, process $p_i$ does a ho-bc for round 1.
    We can now continue our argument in an inductive fashion to conclude that $p_i$ does not crash in the simulated run, which proves the claim.

    By the above claim, every non-faulty $q_i$ receives a non-$\bot$ output from $p_i$ and hence decides.
    Thus protocol $\calQ$ solves task $T$ in $\amp_f$, showing that any task that is solvable in $\cfho_f$ is also solvable in $\amp_f$ for $0 \le f \le n$.

    Finally, suppose that a protocol $\calP$ solves a task $T$ in $\ho_f$.
    We show that $\calP$ can also solve $T$ in $\cfho_f$.
    Consider a  run $\gamma$ of $\calP$ in $\cfho_f$.
    We can consider a run $\gamma'$ of $\calP$ in $\ho_f$ where every process that crashes in $\gamma$, say at the beginning of some round $r$, does not crash in $\gamma'$ but instead from round $r$,
    (1) its messages are still lost to all other processes and (2) it hears from all non-crashed processes in~$\gamma$ (by Weak Liveness, there are at least $n - f$ such processes).
    By definition of $\ho_f$, all processes decide in~$\gamma'$.
    We show that (i) all processes that decide in~$\gamma$ output the same value they decide in~$\gamma'$, meaning that the output vector of~$\gamma$ is valid and (ii) all non-faulty processes in~$\gamma$ decide.
    For (i), if process $p$ crashes in round $r$ in $\gamma$, as $\gamma$ and $\gamma'$ are indistinguishable till round $r$ to $p$, it will decide a value in $\gamma$ if and only if it decides the same value before round $r$ in $\gamma'$.
    If $p$ does not crash in $\gamma$, as $\gamma$ and $\gamma'$ are indistinguishable, it decides in $\gamma$ the same values that it necessarily decides in $\gamma'$.
    This also implies (ii).
\end{proof}

\section{From $\amp_f$ to $\ho_f$} \label{sec:converse}

We now examine the conditions under which tasks solvable in $\amp_f$
are solvable in $\ho_f$.
Instead of directly simulating $\amp_f$ in $\ho_f$,
we introduce the \emph{Silenced-Faulty Heard-Of} ($\sfho_f$) model,
a variant of $\ho_f$, which simplifies our proofs.
The proof strategy for the equivalence results is outlined in Figure~\ref{fig:proof_strategy}.

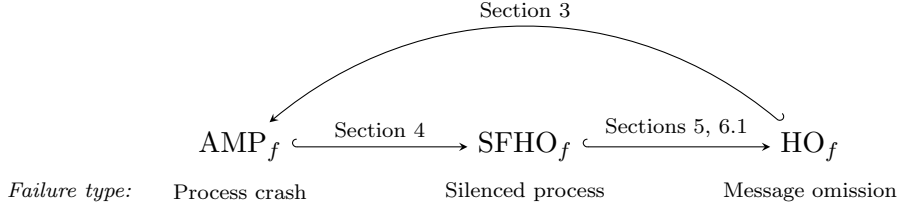
\begin{figure}[tb]
	\centering
	\begin{tikzpicture}[scale=1.5]
		\node (X) at (-1.5, -0.4) {\scriptsize\itshape Failure type:};
		\node (A) at (0, 0) {$\amp_f$};
		\node (A1) at (0, -0.4) {\scriptsize Process crash};
		\node (B) at (2.5, 0) {$\sfho_f$};
		\node (B1) at (2.5, -0.4) {\scriptsize Silenced process};
		\node (C) at (5.0, 0) {$\ho_f$};
		\node (C1) at (5.0, -0.4) {\scriptsize Message omission};
		%\draw[right hook-stealth] (A) edge[bend left] (B);
		%\draw[left hook-stealth] (B) edge[bend left] (A);
		\draw[right hook-stealth] (A) edge node[above] {\scriptsize Section~\ref{sec:converse}} (B);
		\draw[right hook-stealth] (B) edge node[above] {\scriptsize Sections~\ref{subsec:colorless},~\ref{subsec:colored_f=1}} (C);
		\draw[left hook-stealth] (C) edge[bend right=40] node[above] {\scriptsize Section~\ref{sec:HOinAMP}} (A);
	\end{tikzpicture}
	\caption{Strategy for equivalence results; $\mathcal{M}_1 \hookrightarrow \mathcal{M}_2$ denotes a simulation of $\mathcal{M}_1$ in $\mathcal{M}_2$}
	\label{fig:proof_strategy}
\end{figure}

\subsection{The Silenced-Faulty Heard-Of model} \label{sec:sfho}

The schedules in $\sfho_f$ must satisfy exactly the same conditions as they do in $\ho_f$ (Round-robin property, Integrity, No Duplicates, Weak Liveness).
The difference between the models is that we can also declare a process to be faulty in $\sfho_f$.
We first define the notion of a \emph{silenced} process.

The \emph{reach} of a process $p_i$ from round $r$ can be defined as the union of
the set of processes that hear from $p_i$ in round $r$,
the set of processes that hear from one of these processes in round $r + 1$,
the set of processes that hear from one of the previous processes (of rounds $r$ and $r + 1$)
in round $r + 2$ and so on.
Intuitively, the reach captures the transitive closure of message deliveries
originating from $p_i$ until a particular round.
Formally,
\begin{definition}
	\label{defn:reach}
	Let $\textsc{Rcv}_i(r)$ be the set of processes that
	receive a message from process $p_i$ in round $r$.
	The \emph{reach} of process $p_i$ between rounds $r$ and $s \ge r$,
	denoted $\textsc{Reach}_i(r, s)$, is defined as
	\[
		\textsc{Reach}_i(r, s) \stackrel{def}{=}
		\begin{cases}
            \left\{ p_i \right\},                                                                            & \text{if $s = r$} \\
            \bigcup_{p_j \in \textsc{Reach}_i(r, s - 1)} \textsc{Rcv}_j(s - 1), & \text{if $s > r$.}
			%\textsc{Reach}_p(r) \cup \bigcup_{q \in \textsc{Reach}_p(r)} \textsc{Reach}_q(r + 1, s), & \text{if $r < s$}.
		\end{cases}
	\]
	The reach $\textsc{Reach}_i(r, \infty)$ of $p_i$ starting from round $r$ is $\bigcup_{s \ge r} \textsc{Reach}_i(r, s)$.
\end{definition}

\begin{definition}
	A process $p_i$ is \emph{silenced from round $r$}
	if $\abs{\textsc{Reach}_i(r, \infty)} \le f$.
    We say that $p_i$ is \emph{silenced} if there exists $r$ such that $p_i$ is silenced from round $r$.
\end{definition}

\begin{definition}
	\label{defn:faulty_in_sfho}
	A process is \emph{faulty} in a run $\sigma$ in $\sfho_f$ if
	it is silenced in $\sigma$.
\end{definition}

We make the following simple but crucial observation about non-silenced
processes for $n > f$, which also explains the term ``silenced''.
If a process $p_i$ is not silenced, then for any round $r$,
there is a round $s \ge r$ such that $\abs{\textsc{Reach}_i(r, s)}
	= \abs{\textsc{Reach}_i(r, \infty)} \ge f + 1$.
This implies that in round $s+1$
every process hears from some process in $\textsc{Reach}_i(r, s)$,
which implies:

\begin{lemma}
	\label{lem:all_see_non_sil}
    If $n > f$ then $p_i$ is not silenced if and only if for any round $r$,
	$\abs{\textsc{Reach}_i(r, \infty)} = n$.
\end{lemma}

When $n > 2f$, another important observation is that
at most $f$ processes are silenced.

\begin{lemma}
	\label{lem:at_most_f_silenced}
	If $n > 2f$, there are at most $f$ silenced processes.
\end{lemma}
\begin{proof}
	Let $Sil$ be the set of silenced processes.
	Suppose that $p_i \in Sil$ is silenced from round $r_i$.
	Then all the processes in $\textsc{Reach}_i(r_i, \infty)$ must be silenced.
	Indeed, if some process $p_j \in \textsc{Reach}_i(r_i, \infty)$ is not, then by Lemma~\ref{lem:all_see_non_sil}, $\abs{\textsc{Reach}_j(r', \infty)} =~n$ for some $r' > r_i$ and so $\abs{\textsc{Reach}_i(r_i, \infty)} = n$, contradicting our assumption that $p_i$ is silenced.
	%Furthermore, by the definition of reach, there is a round after which no message from any process in $\textsc{Reach}_p(r_p)$ can reach nodes outside $\textsc{Reach}_p(r_p)$.

    By the definition of reach, there is a round after which no message broadcast by a process in $\textsc{Reach}_i(r_i, \infty)$ is received by a process outside the set.
    As we have shown for any $p_i \in Sil$ that $\textsc{Reach}_i(r_i, \infty) \subseteq Sil$, there is a round after which no message broadcast by a process in $Sil$ is received by a process outside~$Sil$.

	We next argue that when $n > 2f$, at least one process is not silenced.
	Otherwise, all processes are silenced by some round $R$.
	In any round $r > R$, a message broadcast by any process is received
	by at most $f$ processes in round $r$.
	So the total number of successfully received messages in round $r$
	is at most $nf$.
	However, by Weak Liveness, this number is at least $n(n - f)$,
	which is a contradiction as $n - f > f$.

	Let $p_k$ be a non-silenced process, that is, $p_k \notin Sil$.
	Since $p_k$ misses at most $f$ messages from other processes
	in any round and eventually does not hear from processes in $Sil$,
	it follows that $\abs{Sil} \le f$;
	that is, there are at most $f$ silenced processes.
\end{proof}

\subsection{From $\amp_f$ to $\sfho_f$} \label{sec:AMPinSFHO}

%{\color{blue}
%  TODO:
%  \begin{itemize}
%    \item Draw attention to the fact that in this AMP model, each amp-recv is followed by at most one amp-send and why this is important.
%  \end{itemize}
%}

\begin{lemma}
	\label{lem:amp_in_sfho}
	If $n > 2f$, then any task that is solvable in $\amp_f$
	is also solvable in $\sfho_f$.
\end{lemma}

\begin{proof}
	Let $T = (\mathcal{I}, \mathcal{O}, \Delta)$ be a task that is
	solvable in $\amp_f$ by a protocol $\calP = \left\{ p_1, \dots, p_n \right\}$.
	%We assume without loss of generality that there is always a send event enabled after every receive event at each $p_i$; this is valid as $\calP$ can be taken to be a full-information protocol for example.
	We use a hard-coded unique identifier~$id_i$ for each process~$p_i$.

	We construct a simulation system that simulates $\amp_f$ in $\sfho_f$ so that for $1 \le i \le n$, (1) the $i$-th simulated process is $p_i$ and (2) the $i$-th simulation machine~$P_i$ has the specification given by Algorithm~\ref{algo:AMPinHO_gen}.
    Roughly, to simulate broadcasting a message in $\amp_f$, it is echoed in every round of $\sfho_f$ by the processes that have already received it, in the hope that all processes eventually receive it.
	This fails when the message's original process is silenced in $\sfho_f$, and the simulated run risks violating the Non-faulty Liveness condition of $\amp_f$.
	So processes rely on acknowledgements to determine if their messages have reached everyone,
	until which they temporarily block themselves from proceeding in $\amp_f$.
	Thus if they are silenced in $\sfho_f$, they are faulty in $\amp_f$, and the Faulty Quasi-Liveness condition allows their last message to be not received by everyone in the simulated run.

	\begin{algorithm}[tb]
		\caption{Pseudocode for machine $P_i$}
		\label{algo:AMPinHO_gen}
		\begin{algorithmic}[1]
			\State Initialize set $seen := \emptyset$ \Comment{History of messages received since the beginning}
            \State Initialize set $old := \emptyset$ \Comment{History of messages received before $p_i$'s last amp-bc}
			\State Initialize variable $latest := \texttt{null}$
			\State Initialize variable $T := 0$ \Comment{For timestamping}
            \State Enable $\operatorname{ho-bc}_i(seen)$ \Comment{Enable broadcast for the first round of $\sfho_f$}
			\medskip
            \State{When $\operatorname{amp-bc}_i(m)$ occurs:}
			\IndState $latest \gets \left\langle id_i, m, T \right\rangle$
			\IndState Add $latest$ to $seen$
			\IndState $T \gets T+1$
			\medskip
			%\State In round $r \ge 1$:
			\State{When $\operatorname{ho-recv}_i(M')$ occurs:}
            \IndState \parbox[t]{403pt}{Add previously unseen messages in $M'$ to $seen$ and acknowledge those of type $\left\langle s, m', t \right\rangle$ by adding $ack(id_i, \left\langle s, m', t \right\rangle)$ to $seen$} \label{algo:update_seen}
            \vspace{2pt}
            \IndState \parbox[t]{423pt}{If $latest \ne  \texttt{null}$ \textbf{and} $ack(j, latest)$ occurs in $M'$ for at least $f$ distinct values of $j \ne id_i$: \label{algo:condition_for_amp_recv}}
			\IndIndState Initialize set $pending := \emptyset$
			\IndIndState Add messages of type $\left\langle s, m', t \right\rangle$ in $seen\!\setminus\!old$ to $pending$
			\IndIndState If $pending$ is empty:
			\IndIndIndState Enable $\operatorname{amp-recv}_i(\bot)$
			\IndIndState Else:
			\IndIndIndState For each $\left\langle s, m', t \right\rangle \in pending$: \Comment{Release all newly received messages} \label{algo:release_messages_sim2}
			\IndIndIndIndState Enable $\operatorname{amp-recv}_i(m')$
			\IndIndState $old \gets seen$
			\IndIndState $latest \gets \texttt{null}$ \Comment{Update $old$ and reset $latest$ for next AMP step}
            \smallskip
            \IndState Enable $\operatorname{ho-bc}_i(seen)$ \Comment{Enable broadcast for the next round of $\sfho_f$}
		\end{algorithmic}
	\end{algorithm}

    To solve a task $T$ in $\sfho_f$, a protocol $\calQ = \left\{ q_1, \dots, q_n \right\}$ can again be constructed in the same manner as in Section~\ref{sec:HOinAMP}, except that now the main events of $q_i$ are ho-bc and ho-recv, while its internal events are amp-bc and amp-recv.

	The simulator $q_i$ works as follows.
	The input of $q_i$ is also used as the input for~$p_i$.
	%When $p_i$ does an internal $\operatorname{amp-send}_i(m)$, $P_i$ adds $\left\langle id_i, m, T \right\rangle$ to its $seen$ set which $q_i$ then broadcasts in every subsequent round of $\sfho_f$.
    When $p_i$ does an internal $\operatorname{amp-bc}_i(m)$, $q_i$ broadcasts $latest = \left\langle id_i, m, T \right\rangle$ for some time value $T$ in every subsequent round of $\sfho_f$.
	Also in every round, it broadcasts all previously received messages and acknowledgements for messages of the form $\left\langle \cdot, \cdot, \cdot \right\rangle$.
	%Every $q_j$ that receives $\left\langle id_i, m, T \right\rangle$ rebroadcasts it each round with an acknowledgement that it received $\left\langle id_i, m, T \right\rangle$.
	%However if $q_i$ is silenced, $\left\langle id_i, m, T \right\rangle$ does not reach everyone, meaning not all $q_j$ enable $\operatorname{amp-recv}_j(m)$ from $P_j$ to $p_j$; the simulated run risks violating the Non-Faulty Liveness condition of $\amp_f$.
	%To avoid this, $q_i$ makes $p_i$ temporarily faulty by blocking all further amp-recv events from $P_i$ to $p_i$ until enough acknowledgements for $\left\langle id_i, m, T \right\rangle$ are received.
	%If this never happens, $p_i$ is faulty and the Faulty Quasi-liveness condition allows the simulated run to be valid.
	%If it is satisfied on the other hand, $q_i$ is certain that everyone will receive $\left\langle id_i, m, T \right\rangle$ and unblocks $p_i$ by enabling an internal $\operatorname{amp-recv}_i(m')$ for every newly received message of the form $\left\langle \cdot, m', \cdot \right\rangle$.
	It waits to receive at least $f$ acknowledgements for~$\left\langle id_i, m, T \right\rangle$ before unblocking $p_i$ by enabling an internal $\operatorname{amp-recv}_i(m')$ for every pending received message of the form $\left\langle \cdot, m', \cdot \right\rangle$.
	If $p_i$ internally decides an output~$o$, then~$q_i$ decides~$o$ as well.

	Every~$P_i$ tags every message~$m$ from~$p_i$ with
	the identifier~$id_i$ and a timestamp~$T$,
    to distinguish an $\amp_f$ message that is broadcast multiple times
	by a process or by different processes.

	%At this point, we can explain the Faulty Quasi-Liveness condition of $\amp_f$.
	%    If $p_i$ is halted by $q_i$, then it is because $q_i$'s $\left\langle id_i, m, T \right\rangle$ message triggered by an $\operatorname{amp-send}_i(m)$ by $p_i$ did not reach enough processes.
	%    Thus not all processes $p_j$ will have an $\operatorname{amp-recv}_j(m)$.
	%    The Faulty Quasi-Liveness condition allows this type of situation. \textcolor{blue}{(rewrite)}

	%As in the previous simulation, we only consider runs of $\calQ$ in which $q_i$ does not halt in a state where there is an amp-send enabled in $p_i$'s internal state or a ho-recv enabled in $P_i$'s internal state.\marginpar{I think this is incomplete}

	%In order to prove that protocol $\calQ$ solves $T$ in $\sfho_f$, consider any run $\gamma$ of $\calQ$ in $\sfho_f$ with input vector $I$ and output vector $O$.
	%As before, we can obtain the underlying simulation schedule for $\gamma$, which we call $\alpha$.
	%%Note that $bot(\alpha)$ is the schedule of events in $\gamma$.
	%Since each $q_i$ in run $\gamma$ uses the output of $p_i$ in $top(\alpha)$ to decide, $O$ is also the vector of outputs decided by the $p_i$ in $top(\alpha)$.
	%We will first show that $(I, O) \in \Delta$ by showing that $top(\alpha)$ is a run in $\amp_f$.

	The first observation is that if~$q_i$ is non-faulty (in $\sfho_f$)
	then~$p_i$ is also non-faulty (in $\amp_f$).
	By Lemma~\ref{lem:all_see_non_sil} and since~$q_i$ is non-silenced,
	every non-null $latest$ value that occurs at~$q_i$
	eventually reaches all processes in~$\calQ$.
	After a $latest$ value by $q_i$ is acknowledged by all the processes, $q_i$ receives at least $n - f > f$ acknowledgements and thus allows $p_i$ to continue in the simulated run.

	Lemma~\ref{lem:at_most_f_silenced} implies that there are
	at most $f$ faulty processes in $\calP$ in the simulated run.

    In the simulated run, if a process $p_i$ has an $\operatorname{amp-bc}(m)$ event eventually followed by an amp-recv event, it means that $q_i$ received at least $f$ acknowledgements to $\left\langle id_i, m, T \right\rangle$ (for some $T$) from others before allowing further amp-recvs at $p_i$.
	Thus eventually all processes $q_j$ receive $\left\langle id_i, m, T \right\rangle$ and as all non-faulty processes $p_j$ have infinitely many amp-recvs, $\operatorname{amp-recv}(m)$ must be one of them.
	Non-faulty Liveness and Faulty Quasi-Liveness hold as a consequence.
	Finally, it is straightforward to show that the simulated run satisfies
	the Integrity and No Duplicates properties.

	We conclude that the simulated run is valid in $\amp_f$, meaning the vector of outputs of the simulated processes $p_i$ and thus that of the $q_i$ is valid.
	Moreover, a non-faulty $q_i$ gets a non-$\bot$ output from its internal $p_i$.
	Thus task $T$ is also solvable in $\sfho_f$, proving the lemma.
\end{proof}

The next sections study the conditions under which a task solvable
in $\sfho_f$ is also solvable in $\ho_f$,
and then apply Lemma~\ref{lem:HOinAMP} and Lemma~\ref{lem:amp_in_sfho}.

\section{Equivalence for colorless tasks} \label{subsec:colorless}

We now show that when $n > 2f$,
$\ho_f$ and $\amp_f$ solve the same set of colorless tasks.
Let $T = (\mathcal{I}, \mathcal{O}, \Delta)$ be a colorless task that
is solvable in $\sfho_f$ by a protocol $\mathcal{P} = \left\{ p_1, \dots, p_n \right\}$.
One can easily construct a protocol $\calQ = \left\{ q_1, \dots, q_n \right\}$
that solves $T$ in $\ho_f$.
In every round, every process $q_i$ runs $p_i$,
and if a yet undecided process $q_i$ hears from process $q_j$ such that
$q_j$ has already decided $o$, then $q_i$ decides $o$.

Let $\alpha$ be a run of $\calQ$ with input vector $I$ and output vector $O_{\calQ}$ in $\ho_f$.
We first prove that all processes of $\calQ$ decide on a value.
\begin{lemma}
	\label{lem:all_processes_decide}
	All processes $q_i$ have non-$\bot$ output values in $O_{\calQ}$.
\end{lemma}
\begin{proof}
	Since $n > 2f$, at least $n - f$ processes in $\calQ$ are not silenced by Lemma~\ref{lem:at_most_f_silenced} and decide by some round.
	In the next round, any $q_i$ that has not yet decided hears from a decided process due to Weak Liveness and copies its output.
\end{proof}

It remains to show that $\calQ$ gives a valid output vector for the input vector of $\alpha$.
Let $O_\calP$ be the output vector of $\calP$ on $\alpha$.
Then $(I, O_\calP) \in \Delta$ as $\alpha$ can be viewed as a run of $\calP$ in $\sfho_f$.
\begin{lemma}
	\label{lem:colorless_property_satisfied}
	$val(O_{\calQ}) \subseteq val(O_{\calP})$ and thus $(I, O_{\calQ}) \in \Delta$.
\end{lemma}
\begin{proof}
	Suppose $q_i$ decides output $o$ in some round.
	This must be because either~$p_i$ decides~$o$ in that round in which case $o \in val(O_{\calP})$, or~$q_i$ heard from some~$q_j$ that had already decided~$o$ in which case we can argue by induction that $o \in val(O_\calP)$.
	Thus $val(O_{\calQ}) \subseteq val(O_{\calP})$ and as~$T$ is colorless, $(I, O_{\calQ}) \in \Delta$.
\end{proof}

Thus protocol $\calQ$ solves colorless task $T = (\mathcal{I}, \mathcal{O}, \Delta)$ in $\ho_f$,
proving that for $n > 2f$, colorless tasks solvable in $\sfho_f$ are solvable in $\ho_f$.
Together with Lemma~\ref{lem:HOinAMP} and Lemma~\ref{lem:amp_in_sfho}, this implies:

\begin{theorem}
	\label{thm:colorless}
	For $n > 2f$, models $\ho_f$ and $\amp_f$ solve the same set of colorless tasks.
\end{theorem}

\section{Equivalence and separation for colored tasks}

\subsection{Equivalence with one fault}
\label{subsec:colored_f=1}

We now show for $f \le 1$ and $n > 2f$ that $\ho_f$ and $\amp_f$
solve the same set of colored tasks.
When $f=0$, there are no silenced processes in $\sfho_0$,
and it is the same model as $\ho_0$ and thus $\amp_0$
(by Lemma~\ref{lem:HOinAMP} and Lemma~\ref{lem:amp_in_sfho}).
Thus, it suffices to consider the case where $f = 1$ and $n > 2f = 2$.

Let $T = (\mathcal{I}, \mathcal{O}, \Delta)$ be a colored task that is solvable
in $\sfho_1$ by a protocol $\calP = \left\{ p_1, \dots, p_n \right\}$.
We show that the following protocol $\calQ = \left\{ q_1, \dots, q_n \right\}$
solves task $T$ in $\ho_1$.
%\textit{Idea:} Once a process $p_i$ is silenced, it receives (nearly) up-to-date views of all the other processes, allowing it to determine if they have decided according to $\mathcal{P}$.
%Once all the others have decided, $p_i$ can consider an identical sequence $\alpha'$ in $\cfho$ where it does not crash, and can choose the value that $\mathcal{P}$ gives it for $\alpha'$.

Recall that unlike colorless tasks,
colored tasks require each process to produce its own output;
thus decisions cannot be freely adopted from other processes.
Thus, while process $q_i$ runs $p_i$ in each round,
the decision rule is modified as in Algorithm~\ref{algo:decision_silenced_node}.

\begin{algorithm}[]
	\caption{Pseudocode for process $q_i$}
	\label{algo:decision_silenced_node}
	\begin{algorithmic}[1]
        \State Round $r$: Broadcast and receive messages and change the state of $p_i$ accordingly
        \State If $q_i$ does not have an output yet, then
		\IndState\label{line:normaldecision} If $p_i$ decides $o$ in round $r$, then $q_i$ decides $o$
		%\IndState Else
        \IndState\label{line:all_but_one_decide} \parbox[t]{373pt}{Else if $q_i$ determines from its round-$r$ view that every other process has decided}
		%\IndIndState\label{line:all_but_one_decide} If all $q_j$ ($j \neq i$) have decided before round $r$ and the round-$r$ view of $q_i$ contains the views of every $q_j$ ($j \ne i$) in the round where it decided, then
		%\IndIndIndState \label{line:constructnewrun}$q_i$ considers an infinite run $\alpha'$ such that for each $q_j$, $\alpha'$ is identical to the current run till the round where $q_i$ last hears from $q_j$, and no process is silenced in $\alpha'$
		%\IndIndIndState\label{line:alternatedecision} $q_i$ decides the output of $p_i$ in $\alpha'$
        \IndIndState \label{line:alternatedecision} \parbox[t]{373pt}{Then $q_i$ chooses an output that together with the others' outputs and the input vector, satisfies the task}
	\end{algorithmic}
\end{algorithm}

Let $\alpha$ be a run of $\calQ$ in $\ho_1$ with input vector $I$
and output vector $O_{\calQ}$.

\begin{lemma}
	\label{lem:can_construct_new_run}
	If $q_i$ executes line~\ref{line:all_but_one_decide} in round $r$,
	then there is a value for $q_i$ to choose in line~\ref{line:alternatedecision}.
\end{lemma}

\begin{proof}
	Process $q_i$ starts constructing a run $\alpha'$ by using the view of each $q_j$ for the round where $q_i$ last hears from $q_j$.
    It can then fill in the missing events in $\alpha'$ till round $r$ by going over all possibilities (there is at least one that is valid, namely run $\alpha$ till round $r$).
	From round $r + 1$ onwards, it assumes that no messages are lost, so no process is silenced in $\alpha'$.
	Thus $p_i$ decides an output in $\alpha'$, which $q_i$ uses in $\alpha$.
	To see that $q_i$'s choice leads to $O_{\calQ}$ being a valid output vector, note that it is equal to the output vector yielded by $\calP$ on $\alpha'$.
\end{proof}

We complement Lemma~\ref{lem:can_construct_new_run} with
the observation that if a process is silenced,
then it has nearly up-to-date information about the entire system.

\begin{lemma}
	\label{lem:silenced_node_hears_from_all}
	If process $q_i$ is silenced from round $R$ in $\alpha$,
	then for every round $r \ge R + 1$,
	its round-$r$ view contains the round-$(r - 2)$ views of all processes.
\end{lemma}
\begin{proof}
	The round-$r$ view of $q_i$ contains the round-$(r - 1)$ views (and hence the round-$(r - 2)$ views) of at least $n - 1$ processes including itself.
	Let $q_j$ be an exception.
	All processes in $\calQ \setminus \left\{ q_i, q_j \right\}$ must hear from $q_j$ in round $r - 1$ as $q_i$ is already silenced and no one hears from it.
	Thus $q_i$ receives the round-$(r - 2)$ view of $q_j$ from $\calQ \setminus \left\{ q_i, q_j \right\}$ in round $r$.
\end{proof}

It remains to show that all processes decide and the run yields a valid output vector.

\begin{lemma}
	\label{lem:colored_termination}
	All processes $q_i$ have non-$\bot$ output values in $O_{\calQ}$.
\end{lemma}
\begin{proof}
	If $q_i$ is not silenced then it is guaranteed to decide using $p_i$.
	Now suppose that $q_i$ is silenced from round $R$ and it does not get to decide using line~\ref{line:normaldecision}.
	By Lemma~\ref{lem:at_most_f_silenced},
	since no other process in $\calQ$ is silenced,
    all processes in $\calQ \setminus \left\{ q_i \right\}$ decide by some round $R'$.
    Thus by Lemma~\ref{lem:silenced_node_hears_from_all}, $q_i$ satisfies the condition in line~\ref{line:all_but_one_decide}
	by round $\max(R, R') + 2$ and also decides,
	by Lemma~\ref{lem:can_construct_new_run}.
\end{proof}

\begin{lemma}
	\label{lem:colored_valid_outputs}
	$(I, O_{\calQ}) \in \Delta$.
\end{lemma}
\begin{proof}
	There can be at most one process $q_i$ that decides using line~\ref{line:alternatedecision}.
	Indeed, if there were two such processes, both must have satisfied line~\ref{line:all_but_one_decide}, meaning each decided their outputs before the other; this is absurd.
	If \emph{all} processes in $\calQ$ decide their outputs in line~\ref{line:normaldecision},
	then $\calP$ must yield output vector $O_{\calQ}$ on run $\alpha$ in $\sfho_1$
	and thus $(I, O_\calQ) \in \Delta$.
	Otherwise, some process $q_i$ decides in line~\ref{line:alternatedecision},
	and thus, all other processes $q_j$ decide in line~\ref{line:normaldecision}.
	By Lemma~\ref{lem:can_construct_new_run},
	$q_i$'s output combined with the others' outputs and the input vector $I$
	satisfy the task and thus $(I, O_\calQ) \in \Delta$.
\end{proof}

Hence protocol $\calQ$ solves task $T = (\mathcal{I}, \mathcal{O}, \Delta)$ in $\ho_1$,
which implies:

\begin{lemma} \label{lem:sfho_in_ho_colored}
	For $f \le 1$ and $n > 2f$, colored tasks solvable in $\sfho_f$ are solvable in $\ho_f$.
\end{lemma}

Combining this with Lemma~\ref{lem:HOinAMP} and Lemma~\ref{lem:amp_in_sfho} implies:

\begin{theorem}
	\label{thm:colored_f=1}
	For $f \le 1$ and $n > 2f$, $\ho_f$ and $\amp_f$ solve the same set of colored tasks.
\end{theorem}

\subsection{Separation with more than one fault}
\label{sec:AMPneHO}

Algorithm~\ref{algo:decision_silenced_node}
cannot be generalized to the case where $f > 1$,
since it is possible to have a run with two silenced processes
such that they are the only ones hearing from themselves.
Since the two silenced processes will never have information
about each other, getting them to decide compatible outputs is a challenge due to symmetry.
Indeed, this prevents $\ho_f$ and $\amp_f$ from coinciding
for colored tasks and $f > 1$.
We demonstrate this with the renaming task~\cite{attiya_renaming_amp_1990}.

The \emph{renaming task} with initial name space of
potentially unbounded size $M$
and new name space of size $N$ and $M > N$ is defined as follows.
Every process initially has as input a distinct identifier from the initial name space.
Every process with a non-$\bot$ output must return a distinct output name from the new name space.
To avoid trivial protocols for this task (e.g., for $N = n$, the protocol where process $p_i$ always chooses $i$ as output), we use the \emph{anonymity assumption} on renaming protocols used in~\cite{attiya_renaming_amp_1990}.
In particular, we can assume that the initial state of a process is a function of only its initial name and that all processes have the same decision function $\delta$ that maps the sequence of views at a process to an output name.

The anonymity assumption states that for any pair of processes and any input name $i$, both the processes have the same initial state when given $i$ as input.
Furthermore, let $\pi$ be a permutation of $\left\{ 1, \dots, n \right\}$ and let $J$ and $J'$ be two initial configurations such that for any $1 \le i \le n$, processes $p_i$ in $J$ and $p_{\pi(i)}$ in $J'$ have the same state.
Let $R$ be a run of the renaming protocol starting from $J$.
Consider $\pi(R)$, also a run of the protocol, starting from $J'$ such that the sequence of events executed by $p_i$ in $R$ is instead executed by $p_{\pi(i)}$ in $\pi(R)$ for all $i$.
Then $p_i$ in $R$ and $p_{\pi(i)}$ in $\pi(R)$ must have the same sequence of states.

\begin{theorem}
	\label{thm:colored_f>1}
	For $1 < f < n/2$, the renaming task with an initial name space of size $N + n - 1$ 
    and new name space of size $N = n + f$ is solvable in $\amp_f$ 
    but is not solvable in $\ho_f$.
\end{theorem}

\begin{proof}
	For $n > 2f$, the renaming task with an unbounded initial name space
	and a new name space of size $N = n + f$ is solvable in
	$\amp_f$~\cite[Theorem 7.9]{attiya_renaming_amp_1990}.
	Trimming the size of the initial name space to $N + n - 1$ clearly preserves the result.
	It is worth noting that the specifications of the $\amp_f$ model in \cite{attiya_renaming_amp_1990} differ in that all messages sent to non-faulty processes are required to be received.
	However it can be checked that the protocol works even if
	we weaken this by assuming Faulty Quasi-Liveness.

	To show the separation, assume by way of contradiction
	that there is a protocol $\calP$ solving the task in $\ho_f$.
	We show how to obtain an input vector on which $\calP$ fails.
	Fix arbitrary distinct initial names $x_1, \dots, x_{n - 2}$ for processes $p_1, \dots, p_{n - 2}$.
	Let processes $p_{n - 1}$ and $p_n$ have initial names $a$ and $b$ that will be chosen later.
	Consider a run of $\calP$ in $\ho_f$ with input vector $(x_1, \dots, x_{n - 2}, a, b)$ such that: (1) $p_1, \dots, p_{n - 2}$ hear from each other in every round, (2) $p_{n - 1}$ and~$p_n$ are not heard of by any other process in any round, and (3) $p_{n - 1}$ and $p_n$ get to hear from $p_1, \dots, p_{n - 2}$ in every round.

	With $x_1, \dots, x_{n - 2}$ fixed, the sequence of views of $p_{n - 1}$ is the value of some function $f(a)$ that does not depend on $b$, and its new name is some value $\delta(f(a))$.
	The sequence of views of $p_{n - 1}$ and $p_n$ are symmetric, so the sequence of views of $p_{n}$ is $f(b)$ and its new name is $\delta(f(b))$.
	As there are $N + 1$ initial names apart from $x_1, \dots, x_{n - 2}$ and only $N$ possible new names, there exist values for $a, b$ such that such that $x_1, \dots, x_{n - 2}, a, b$ are distinct and $\delta(f(a)) = \delta(f(b))$.
	Thus~$\calP$ fails on the input vector $(x_1, \dots, x_{n - 2}, a, b)$; a contradiction.
\end{proof}

\section{Randomized protocols with a non-adaptive adversary} \label{sec:randomized}

In this section, we extend our study to randomized protocols under
a \emph{non-adaptive adversary},
which fixes process crashes and message omissions independently
of the protocol’s random choices.
Randomization can break symmetry among processes, but it cannot repair lost communication:
once two processes are silenced with respect to each other,
their random choices evolve independently and their outputs may diverge.
We prove that the equivalence between $\amp_f$ and $\ho_f$
for colorless tasks persists in this setting,
while the separation for colored tasks remains.
Thus, the expressive boundary between the two models is \emph{structural},
determined by the flow of information rather than by deterministic constraints.

%We now examine the results of Theorems~\ref{thm:colorless},~\ref{thm:colored_f=1},~\ref{thm:colored_f>1} in the context of randomized protocols.
More specifically, we assume processes are equipped with their own local coins,
which they may flip an arbitrary number of times during the local computation of each step.
Together with an initial configuration and schedule, the sequences of local coin flips determine a run.
An overview of randomized protocols can be found in~\cite{Aspnes03AsyncCons}.

We assume that the scheduling of processes and messages is handled by
a \emph{non-adaptive adversary}, which must fix the schedule of the run in advance.
Once fixed, the results of the coin flips are the only remaining unknowns
that determine the run of the protocol.
%This is in contrast to an \textit{adaptive} adversary which at every point can see the entire history of the system including outcomes of past coin flips, internal states of processes and the contents of messages, and accordingly influence the next event.
We also alter the definition of task solvability in Section~\ref{sec:models}
by replacing the second condition that says that all non-faulty processes must decide.
Instead, we require that once the adversary fixes the scheduling,
all non-faulty processes must decide with probability 1.
The first condition remains unchanged, i.e.,
the output vector of a run always has to satisfy the task specification.

The following randomized versions of
Theorems~\ref{thm:colorless} and~\ref{thm:colored_f=1}
are proved by identical simulation constructions,
detailed in Appendix~\ref{randomized_appendix}.
The only difference is in showing almost-sure termination.

\begin{theorem}\label{thm:colorless_rand}
	For $n > 2f$, $\amp_f$ and $\ho_f$ solve the same set of randomized colorless tasks.
\end{theorem}

\begin{theorem}\label{thm:colored_f=1_rand}
	For $f \le 1$ and $n > 2f$, $\amp_f$ and $\ho_f$ solve the same set of randomized colored tasks.
\end{theorem}

In order to show the randomized counterpart of the separation result
(Theorem~\ref{thm:colored_f>1}),
we will use a probabilistic version of the pigeonhole principle
that shows a positive probability of name collisions of randomized renaming
protocols in~$\ho_f$ (proved in Appendix~\ref{randomized_appendix}).

\begin{lemma}\label{lem:randomized-pigeonhole}
	Let~\(N\ge 1\) be a positive integer and let \(0<c<{1}/{N}\).
	There exists an integer \(L=L(N,c)\) with the following property:
	for every collection \(X_1,\dots,X_L\) of~\(L\) pairwise independent random variables in the set \(\{a_1,\dots,a_N\}\), there exist distinct indices \(a,b \in \{1,\dots,L\}\) such that
	%    \begin{equation}
	$\mathbb{P}(X_a = X_b) \geq c$.
	%    \end{equation}
\end{lemma}

To demonstrate separation, we will use the renaming task as in the deterministic version, only with a potentially larger initial name space because of Lemma~\ref{lem:randomized-pigeonhole}.

\begin{theorem}
	\label{thm:colored_f>1:randomized}
	For $1 < f < n/2$, there exists $L \ge 2$ such that the renaming task with initial name space of size $n - 2 + L$ and new name space of size $N = n + f$ is solvable in $\amp_f$ but not in $\ho_f$.
\end{theorem}
\begin{proof}
	The task is solvable for any $L \ge 2$ in $\amp_f$ even by a deterministic protocol.
	Assume for all $L \ge 2$ that there is a randomized protocol solving it in $\ho_f$ against a non-adaptive adversary.

	We choose the same message schedule and input vector $(x_1, \dots, x_{n - 2}, a, b)$ as in the proof of Theorem~\ref{thm:colored_f>1}, with $a$ and $b$ to be chosen later.
	Let~$R$ be the set of runs having the chosen input vector and schedule.
	%Let $\tilde{R} \subseteq R$ be the set of runs in which all processes decide.
	%By assumption, we have $\mathbb{P}(\tilde{R}) = 1$.
	In run $r\in R$, let~$c'(r)$ be the sequence of random choices of the processes $p_1, \dots, p_{n-2}$ and let $c_{n - 1}(r)$ and $c_n(r)$ be those of $p_{n - 1}$ and $p_n$.
	Write~$\mathcal{C}'$ for the set of sequences of random choices~$c'(r)$.
	%The new names of $p_{n - 1}$ and $p_n$ can then be written as $f(a, c_{n - 1}(r), c'(r))$ and $f(b, c_{n}(r), c'(r))$\todo{explain why this is well defined (like in the deterministic case)}.
	Once $x_1, \dots, x_{n - 2}$ are fixed, the new names of $p_{n - 1}$ and $p_n$ depend on their own initial names, their own sequences of random choices and those of $p_1, \dots, p_{n - 2}$, and thus can be written as $\delta(a, c_{n - 1}(r), c'(r))$ and $\delta(b, c_{n}(r), c'(r))$.
	%Further, write~$f(a,c')$ for the new name of processes~$p_{n-1}$ or~$p_n$ in run~$r$ when starting with initial name~$a$ and observing random choices~$c'$ from processes $p_1,\dots,p_{n-2}$.
	We use the fact that all processes almost-surely decide distinct non-$\bot$ values and the law of total probability~\cite[Theorem 34.4]{Billingsley1995Prob} to obtain
	\begin{align}\label{eq:total:prob}
		1 & = \mathbb{P}\big( \delta(a, c_{n - 1}(r) ,c'(r)) \neq \delta(b, c_n(r), c'(r)) \big) \nonumber                                    \\
		  & = \int_{\mathcal{C}'} \mathbb{P}(\delta(a, c_{n - 1}(r), c'(r)) \neq \delta(b, c_n(r), c'(r)) \mid c'(r) = c')\, d\mathbb{P}(c').
	\end{align}
	%\begin{equation}
	%    1
	%    =
	%    \mathbb{P}(\tilde{R})
	%    =
	%    \mathbb{P}\big( f(a, c_{n - 1}(r) ,c'(r)) \neq f(b, c_n(r), c'(r)) \big)
	%    =
	%    \int_{\mathcal{C}'} \mathbb{P}(f(a, c_{n - 1}(r), c'(r)) \neq f(b, c_n(r), c'(r)) \mid c'(r) = c')\, d\mathbb{P}(c')
	%\end{equation}
	Once $c'(r)$ is fixed to $c'$ and values for $a$ and $b$ are chosen, the new names $\delta(a, c_{n - 1}(r), c')$ and $\delta(b, c_n(r), c')$ are independent random variables because of $c_{n - 1}(r)$ and $c_n(r)$ being independent.
	Moreover they must be chosen from a set of $k = N - (n - 2)$ names.
	%Let $k = N - (n-2)$ the set of possible new names for~$p_{n-1}$ and~$p_n$.
	By setting $c = 1/2k$ in Lemma~\ref{lem:randomized-pigeonhole}, there exists a choice of initial names $a$ and $b$ from a set of $L = L(N, c)$ names such that
	%As $c_{n - 1}(r)$ and $c_n(r)$ are independent, so are $f(a, c_{n - 1}(r) ,c'(r))$ and $f(b, c_n(r), c'(r))$.
	%By Lemma~\ref{lem:randomized-pigeonhole}, setting $c = 1/2k$, since the choice of new names of~$p_{n-1}$ and~$p_n$ are independent, there exists a choice of initial names~$a$ and~$b$ such that we have:
	\begin{equation}\label{eq:conflict:prob}
		\mathbb{P}(\delta(a, c_{n - 1}(r), c'(r)) \neq \delta(b, c_{n}(r), c'(r)) \mid c'(r) = c') \leq 1 - \frac{1}{2k}.
	\end{equation}
	Combining~\eqref{eq:total:prob} and~\eqref{eq:conflict:prob} gives
	\begin{equation}
		1
		\leq
		\left( 1 - \frac{1}{2k} \right) \int_{\mathcal{C}'} d\mathbb{P}(c')
		=
		1 - \frac{1}{2k}
		<
		1
		\enspace,
	\end{equation}
	a contradiction.
\end{proof}

\section{Summary and Discussion}  \label{sec:disc}

%This work establishes a precise correspondence between the asynchronous message-passing model with crash failures, $\amp_f$, and the round-based Heard-Of model with message omissions, $\ho_f$.
This work establishes a precise correspondence between $\amp_f$, the asynchronous message-passing model with crash failures, and $\ho_f$, the round-based Heard-Of model with message omissions.
The two coincide for colorless tasks when $n > 2f$, and for colored tasks only when $f = 1$ (and $n > 2$).
The distinction stems from \emph{silenced processes} in $\ho_f$:
when several processes become permanently unheard, their views diverge and they may reach incompatible decisions.
With a single such process, compatibility can still be preserved; with more, it cannot.
The results extend to randomized algorithms with non-adaptive adversaries, indicating that the limitations arise from communication structure rather than determinism.

The notion of silencing offers a new perspective on the limits of
coordination in asynchronous systems, showing that failures of propagation,
rather than timing, define what is computable.
One open question involves understanding the case $n \leq 2f$.
While it initially appears that the landscape of solvable problems in both $\amp_f$ and $\ho_f$ might only contain trivial tasks,
it turns out that $(n - 1)$-set agreement is solvable in both for $f = n - 2$:
this is done by doing one round of exchanging input values and deciding on the minimal value received.
Other open directions include identifying broader task classes for which the models coincide,
and clarifying the role of the intermediate model $\sfho_f$
for colored tasks with $f > 1$ and $n > 2f$.
More broadly, reasoning in terms of information propagation
may help reveal similar boundaries in models of
partial synchrony and Byzantine behavior.

%====
%
%We have shown that in the context of solvability of colorless tasks, $\amp_f$ and $\ho_f$ coincide for $f < n/2$.
%For colored tasks, this is true if $f = 1$ (and $n > 2$), and false if $1 < f < n/2$.
%The reason for this split in the case of colored tasks is explained by the presence of silenced processes in $\ho_f$.
%For $f = 1$, there is at most one silenced process; it will eventually learn the outputs of the other processes and accordingly decide an output that is compatible.
%For $f > 1$, there can be two such processes who are never heard of by any other process; this makes it difficult for them to decide compatible outputs.
%
%We have not yet addressed the question for the case where $f \ge n/2$.
%We can also ask if there is a set of tasks larger than that of colorless tasks for which the two models coincide.
%A comparison of $\amp_f$ and $\sfho_f$ would also be interesting.
%Indeed, our results show that like $\ho_f$, the $\sfho_f$ model coincides with $\amp_f$ for colorless tasks where $f < n/2$ and colored tasks where $f = 1$.
%However it is not clear whether they agree or differ for the case of colored tasks where $1 < f < n/2$.

%% [Mention the extension of our result for colorless tasks to weak symmetry tasks]

\subsubsection*{Acknowledgements} Dhrubajyoti Ghosh and Thomas Nowak are supported by the ANR projects DREAMY (ANR-21-CE48-0003) and COSTXPRESS (ANR-23-CE45-0013).
Hagit Attiya is supported by the Israel Science Foundation (grant number 25/1849).
Armando Casta\~neda is supported by research projects DGAPA-PAPIIT IN108723 and IN103126, and SECIHTI CBF-2025-I-393.

%%fakesection Bibliography
\bibliography{refs}

\appendix
\section{Proofs for Section~\ref{sec:randomized}} \label{randomized_appendix}

We first prove Theorems~\ref{thm:colorless_rand} and~\ref{thm:colored_f=1_rand}.
\begin{theorem*}
    For $n > 2f$, $\amp_f$ and $\ho_f$ solve the same set of randomized colorless tasks.
\end{theorem*}
\begin{theorem*}
    For $f \le 1$ and $n > 2f$, $\amp_f$ and $\ho_f$ solve the same set of randomized colored tasks.
\end{theorem*}
%\begin{proof}[Proofs of Theorems~\ref{thm:colorless_rand} and~\ref{thm:colored_f=1_rand}]
    The randomized versions of Lemmas~\ref{lem:HOinAMP} and~\ref{lem:amp_in_sfho}
    can be proved using the same simulations:
    %, as we briefly explain.
    Every simulated process $p_i$ now uses randomization and the processes $q_i$ are constructed as before.
    The non-adaptive adversary fixes the scheduling of $\calQ$, which also fixes that of $\calP$.
    We can use the same proofs to show that any simulated run of $\calP$ is valid and obtain that any output vector yielded by $\calP$ is valid and that all non-faulty $p_i$ decide with probability 1.
    It follows that any output vector yielded by $\calQ$ is also valid and moreover as it is still true that $p_i$ is non-faulty if $q_i$ is, all non-faulty $q_i$ decide with probability 1.

    To prove the randomized version of the result in Sections~\ref{subsec:colorless} and~\ref{subsec:colored_f=1},
% Lemmas~\ref{lem:sfho_in_ho_colorless} and~\ref{lem:sfho_in_ho_colored}, 
    the same protocol constructions work:
    %, as we now explain.
    The proofs of the output vector $O_\calQ$ being valid (Lemmas~\ref{lem:colorless_property_satisfied} and~\ref{lem:colored_valid_outputs}) remain unchanged, and we only need to show the randomized versions of Lemmas~\ref{lem:all_processes_decide} and~\ref{lem:colored_termination}.
    In both the lemmas, if $q_i$ is not silenced, its probability of deciding is at least that of $p_i$, i.e.\, 1.
    Thus $n - f$ processes of $\calQ$ decide in a run with probability 1.
    Then using the same ideas in Lemmas~\ref{lem:all_processes_decide} and~\ref{lem:colored_termination}, we conclude that non-faulty processes in $\calQ$ decide with probability 1.
%\end{proof}

We next prove Lemma~\ref{lem:randomized-pigeonhole}.
\begin{lemma*}
    Let~\(N\ge 1\) be a positive integer and let \(0<c<{1}/{N}\).
    There exists an integer \(L=L(N,c)\) with the following property:
    for every collection \(X_1,\dots,X_L\) of~\(L\) pairwise independent random variables taking values in the set \(\{a_1,\dots,a_N\}\), there exist distinct indices \(a,b \in \{1,\dots,L\}\) such that
%    \begin{equation}
        $\mathbb{P}(X_a = X_b) \geq c$.
%    \end{equation}
\end{lemma*}

\begin{proof}
    We will view the random variables~\(X_i\) as vectors \(x=(x_k)_{1\leq k\leq N}\) in the standard simplex \(\Delta^{N-1}\subseteq\mathbb{R}^N\) with \(x_k = \mathbb{P}(X_i = a_k)\).
    We then need to show that there are distinct vectors~\(x\) and~\(y\) such that we have following inequality on their dot product:
    \begin{equation}
        x\cdot y=\sum_{k=1}^N x_k y_k \ge c
    \end{equation}

    Every vector \(x\in\Delta^{N-1}\) satisfies
    \(\lVert x\rVert_2^2=\sum_{k=1}^N x_k^2 \ge {1}/{N}\), since \((\sum_{k=1}^N x_k)^2 \le N\sum_{k=1}^N x_k^2\) by the Cauchy--Schwarz inequality.

    Suppose by contradiction that for some \(c<{1}/{N}\) there is an arbitrarily large finite
    collection of vectors with the property that for all distinct vectors \(x,y\), we have
    \begin{equation}\label{eq:inner-small}
        x\cdot y < c.
    \end{equation}
    Then for any distinct~\(x\) and~\(y\), we have
    \begin{equation}
        \lVert x-y\rVert_2^2 = \lVert x\rVert_2^2 + \lVert y\rVert_2^2 - 2\,x\cdot y
        > 2\Big(\frac{1}{N}-c\Big) > 0.
    \end{equation}
    Hence all vectors are separated by the positive distance
    \begin{equation}
        d = \sqrt{2\big(\tfrac{1}{N}-c\big)}.
    \end{equation}
    But the simplex \(\Delta^{N-1}\) is compact, so any \(d\)-separated subset of \(\Delta^{N-1}\) is finite. Therefore there is a uniform finite upper bound \(M=M(N,c)\) on the size of any set satisfying \eqref{eq:inner-small}. Taking \(L=M+1\) proves the existence of the claimed finite \(L=L(N,c)\).
\end{proof}

A crude explicit packing-type bound for $L$ may be obtained as follows.
The simplex is contained in a Euclidean ball of radius at most \(\sqrt{N}\),
so the maximal number \(M(N,c)\) of points pairwise at distance at least \(d\) satisfies (very roughly)
\[
    M(N,c) < \left(1+\frac{2\sqrt{N}}{d}\right)^{\,N-1}
    = \left(1+\frac{2\sqrt{N}}{\sqrt{2(1/N-c)}}\right)^{N-1},
\]
and one may take \(L=M(N,c)\).

\section{Compactness of $\ho_f$} \label{sec:compactness}

In this section, we prove that the set of runs of a protocol in $\ho_f$ with a \textit{finite} number of possible initial configurations is compact with respect to the \textit{longest-common-prefix} metric.
Given two sequences $\alpha = (\alpha_k)_{k\in\mathbb{N}_0}$ and $\beta = (\beta_k)_{k\in\mathbb{N}_0}$, their \emph{longest-common-prefix distance}~\cite{AS84} is:
\begin{equation}
d(\alpha,\beta)
=
2^{-\inf\{ k \in \mathbb{N}_0 \mid \alpha_k \neq \beta_k \}}.
\end{equation}
If~$\alpha$ and~$\beta$ are identical, then the infimum is $+\infty$ and their distance is zero.
If not, then their distance is equal to~$2^{-K}$ where~$K$ is the smallest index at which the sequences differ.
%In order to equip the set of runs with this metric, we use the alternative equivalent definition of a run where it is a sequence alternating between configurations and process events.

We first define schedules using \textit{communication graphs} and show that the set of schedules is compact when equipping it with the longest-common-prefix metric.
We then show that the mapping from schedules to protocol runs in $\ho_f$ is continuous.
This then proves that the set of runs is compact, as the continuous image of a compact set.

We define~$\mathcal{G}_{n,f}$ as the set of directed graphs $G = (V,E)$ such that:
\begin{enumerate}
    \item the set of vertices of~$G$ is equal to $V = [n]$,
    \item every vertex has a self-loop in $G$, i.e., $(i, i)\in E$ for all $i\in [n]$,
    \item the set of incoming neighbors of every vertex has size at least $n-f$, i.e., $\lvert \{ j \in [n] \mid (j,i) \in E  \} \rvert \geq n - f$ for every $i \in [n]$.
\end{enumerate}
A \emph{communication graph} is a graph $G \in \mathcal{G}_{n,f}$.
A \emph{schedule} is an infinite sequence $G_1,G_2,\dots$ of communication graphs.

To prove compactness of the set of schedules, we use Tychonoff's theorem
(see, \eg~\cite[Chapter~I, \S9, no.~5, Theorem~3]{Bourbaki:top}).

\begin{theorem}[Tychonoff]\label{thm:tychonoff}
    The product of a family of compact topological spaces is compact with respect to the product topology.
\end{theorem}

\begin{lemma}\label{lem:schedules:are:compact}
    The set of schedules is compact with respect to the longest-common-prefix metric.
\end{lemma}
\begin{proof}
    The set of schedules is equal to $\mathcal{G}_{n,f}^\mathbb{N}$.
    It is straightforward to show that the product topology is induced by the longest-common-prefix metric if every copy of~$\mathcal{G}_{n, f}$ is equipped with the discrete metric.
    Each~$\mathcal{G}_{n, f}$ is compact since it is finite.
    Thus, applying Theorem~\ref{thm:tychonoff} shows that the set of schedules is compact as well.
\end{proof}

%We now restrict the set of all schedules to the set of allowable schedules, which will be able to be mapped to allowable sequences for $\ho$.
%A schedule of communication graphs~$G_r$ is \emph{allowable} if crashed processes don't recover: $F(G_{r+1}) \supseteq F(G_r)$ for all $r \in \mathbb{N}$.
%
%\begin{lemma}\label{lem:allowable:schedules:are:compact}
%    The set of allowable schedules is compact with respect to the longest-common-prefix metric.
%\end{lemma}
%\begin{proof}
%    Every closed set in a compact space is compact \cite[Chapter~I, \S9, no.~3, Proposition~3]{Bourbaki:top}.
%    It hence suffices to prove that the set of allowable schedules is closed in the set all schedules.
%    But this is the case since the defining property of allowable schedules is a safety property~\cite{AS84}.
%\end{proof}

Given a protocol in $\ho_f$, denoting by~$\Sigma$ the set of 
schedules, by~$\mathcal{S}$ the set of runs of $\ho_f$, and
by~$\mathcal{C}_0$ the set of initial configurations of the protocol, we
define the mapping $f:\mathcal{C}_0 \times \Sigma \to \mathcal{S}$ inductively 
on every element $(C_0, \alpha) \in \mathcal{C}_0 \times \Sigma$
by simulating the protocol starting from the given initial configuration
$C_0$ and delivering messages in round~$r$ according to the
round-$r$ communication graph~$G_r$ of schedule $\alpha$.

\begin{lemma}\label{lem:scheduler:is:continuous}
    The function~$f\colon \Sigma \times \mathcal{C}_0 \to \mathcal{S}$ is continuous when equipping~$\Sigma$ and~$\mathcal{S}$ with the longest-common-prefix metric and $\mathcal{C}_0$ with the discrete metric.
\end{lemma}
\begin{proof}
    Since the round-$r$ prefix of $f(\alpha, C_0)$ is entirely determined by~$C_0$ and the round-$r$ prefix of~$\alpha$, we have $d(f(\alpha,C_0), f(\beta,C_0)) \leq d(\alpha, \beta)$ for all $\alpha,\beta\in\Sigma$ and all $C_0\in\mathcal{C}_0$.
\end{proof}

\begin{lemma}\label{lem:ho:is:compact}
    For a protocol $\calP$ whose set~$\mathcal{C}_0$ of initial configurations is finite,
    the set of runs of $\calP$ in $\ho_f$ is compact with respect to the longest-common-prefix metric.
\end{lemma}

\begin{proof}
    It is straightforward to check that the set of runs of $\calP$ is the image of the function~$f$.
    We thus have $\mathcal{S} = \bigcup_{C_0 \in \mathcal{C}_0} f(\Sigma, C_0)$.
    Each set $f(\Sigma, C_0)$ is compact as the continuous image of a compact set.
%\cite[Chapter~I, \S9, no.~4, Theorem~2]{Bourbaki:top}.
Then, if~$\mathcal{C}_0$ is finite, the set~$\mathcal{S}$ is compact as the finite union of compact sets.
% \cite[Chapter~I, \S9, no.~3, Proposition~5]{Bourbaki:top}.
%%HA: I don't think these facts need a reference. they are standard.
\end{proof}

\section{$\amp$ is equivalent to the FLP definition} \label{sec:AMP=FLP}

In the FLP version of the $\amp_f$ model, we denote the communication primitives by flp-send and flp-recv.
In a send event $\operatorname{flp-send}_i(m, j)$, process $p_i$ sends message $m$ to a \emph{single} process $p_j$.
In a receive event, as in $\amp_f$, the process receives either a single message or nothing.
%Every flp-recv event is followed by at most one flp-send event.
As in $\amp_f$, runs in $\flp_f$ may have at most $f$ faulty processes and must satisfy the Integrity and No Duplicates properties.
Importantly, we have a stronger condition than Faulty Quasi-Liveness where we require that non-faulty processes receive all messages sent to them.

In order to show that our definition of $\amp_f$ is equivalent to the $\flp_f$ definition, we construct simulations between the two models (see Section~\ref{sec:models} for the definition of a simulation).
For the simulation of $\amp_f$ in $\flp_f$,
the broadcasting of a message in $\amp_f$ can be simulated in $\flp_f$
by sending it individually to all processes, and a message $m$ is said 
to be received in $\amp_f$ if $m$ is received in $\flp_f$.

\begin{algorithm}[H]
    \caption{Machine $P_i$ simulating $\amp_f$ in $\flp_f$}
    \begin{algorithmic}[1]
        \State{When $\operatorname{amp-bc}_i(m)$ occurs:}
        \IndState Enable $\operatorname{flp-send}_i(m, j)$ for all $j$
        \medskip
        \State{Enable $\operatorname{amp-recv}_i(m)$ when:}
        \IndState $\operatorname{flp-recv}_i(m)$ occurs
    \end{algorithmic}
\end{algorithm}

For the simulation of $\flp_f$ in $\amp_f$, the sending of a message $m$ in $\flp_f$ by $p_i$ to $p_j$ can be simulated in $\amp_f$ by having $P_i$ broadcasting the message $\left\langle m, j \right\rangle$, and a message $m$ is said to be received by process $p_i$ in $\flp_f$ if $P_i$ receives $\left\langle m, i \right\rangle$ in $\amp_f$.

%{\color{red}
%    Suppose that a task $T$ is solvable in $\flp_f$ by some protocol $\calP = \left\{ p_1, \dots, p_n \right\}$.
%    Consider the simulation system where (i) the $i$-th simulated process is $p_i$ and (ii) the $i$-th simulation machine is $P_i$, specified in Algorithm~\ref{algo:flp_in_amp}.
%}
\begin{algorithm}[H]
    \caption{Machine $P_i$ simulating $\flp_f$ in $\amp_f$}
    \label{algo:flp_in_amp}
    \begin{algorithmic}[1]
        \State{When $\operatorname{flp-send}_i(m, j)$ occurs:}
        \IndState Enable $\operatorname{amp-bc}_i(\left\langle m, j \right\rangle)$
        \medskip
        \State{Enable $\operatorname{flp-recv}_i(m)$ when:}
        \IndState $\operatorname{amp-recv}_i(\left\langle m, i \right\rangle)$ occurs
    \end{algorithmic}
\end{algorithm}

%{\color{red}
%    A simulated run $\alpha$ in the simulation system of Algorithm~\ref{algo:flp_in_amp} can be prevented from being valid in $\flp_f$.
%    Indeed, it is possible that $p_i$ does $\operatorname{flp-send}_i(m, j)$, triggering $P_i$ which does $\operatorname{amp-bc}_i(\left\langle m, j \right\rangle)$ but then $P_i$ crashes and $\left\langle m, j \right\rangle$ does not reach $P_j$.
%    In this case, $p_j$ would not do $\operatorname{flp-recv}_i(m)$.
%
%    Fortunately this does not prevent non-faulty processes in $\alpha$ from deciding since for them, $\alpha$ is indistinguishable from a run $\alpha'$ that is nearly identical to $\alpha$ except that simulated processes like $p_i$ crash before doing the $\operatorname{flp-send}_i(m, j)$.
%    Note that $\alpha'$ would be a valid run $\flp_f$.
%}

Note that in a simulated run $\alpha$ of the system described by Algorithm~\ref{algo:flp_in_amp}, the last message $m$ sent by a faulty process $p_i$ to a non-faulty process $p_j$ is not necessary received by $p_j$, preventing $\alpha$ from being valid in $\flp_f$.
This is because of the Faulty Quasi-Liveness property of $\amp_f$ whereby $P_i$'s last broadcast $\left\langle m, j \right\rangle$ is not necessarily received by all machines, in particular $P_j$.
Fortunately this does not prevent non-faulty processes in $\alpha$ from deciding, since for them, $\alpha$ is indistinguishable from a valid $\flp_f$ run $\alpha'$ that is almost identical to $\alpha$ except that faulty simulated processes crash before they can do the incomplete send, e.g., $\operatorname{flp-send}_i(m, j)$ in the case of $p_i$.

\end{document}